\documentclass[journal]{IEEEtran}
\usepackage{comment}

\usepackage{geometry}
\usepackage[english]{babel}
\usepackage{amsthm}
\usepackage{amsmath}
\usepackage{graphicx}
\newtheorem{lemma}{Lemma}
\newtheorem{theorem}{Theorem}
\newtheorem{example}{Example}
\renewcommand{\k}{\mathbf{k}}
\newcommand{\T}{^\mathrm{T}}
\newcommand{\delt}{\tilde{\boldsymbol{\Delta}}}
\newcommand{\de}{\boldsymbol{\Delta}}
\newcommand{\cb}{\mathbf{c}}
\newcommand{\kt}{\tilde{\mathbf{k}}}
\newcommand{\kx}{k_\mathrm{x}}
\newcommand{\ky}{k_\mathrm{y}}

\begin{document}

\title{Spatially Selective Reconfigurable Intelligent Surfaces Through Element Permutation}

\author{Fredrik Rusek, Jose Flordelis, Erik Bengtsson, Kun Zhao, Olof Zander\\
       Sony Europe, Lund, Sweden. Contact author: fredrik.x.rusek@sony.com

 \thanks{This research is supported  by the HORIZON-JU-SNS-2022-STREAM-B-01-03 6G-SHINE project (Grant Agreement No. 101095738).}
 }
\maketitle
\thispagestyle{empty}
\begin{abstract}
A standard reconfigurable intelligent surface (RIS) can be configured to reflect signals from an arbitrary impinging direction to an arbitrary outgoing direction. However, if a signal impinges from any other direction, said signal is reflected, with full beamforming gain, to a specific direction, which is easily determined. The goal of this paper is to propose a RIS which \emph{only} reflects signals from the configured impinging direction. This can be accomplished by a RIS architecture that permutes the antenna elements in the sense that a signal is re-radiated from a different antenna than the one receiving the signal. We analytically prove this fact, and also discuss several variants and hardware implementations.
\end{abstract}

\begin{IEEEkeywords}
Reconfigurable intelligent surface (RIS), spatial selectivity.
\end{IEEEkeywords}


\IEEEpeerreviewmaketitle
\section{Introduction}
\IEEEPARstart{R}{econfigurable} intelligent surfaces (RISs) is today a well-established technology, and most theoretical aspects have received major research attention; see, for example,~\cite{RISbook, SmartRadio} and references therein. Recently, the volume of available literature on how RIS technology may be integrated in contemporary and emerging wireless systems has been growing; see, for example,~\cite {industryview}. Further, the European Telecommunications Standards Institute (ETSI) has a currently ongoing industry specification group (ISG) on RIS, which focuses on pre-standardization work, and several group reports are publicly available~\cite{ETSI_RIS}. A technology related to RIS is, arguably, that of Network Controlled Repeaters (NCRs), which will be a part of 3GPP's Release-18~\cite{NCR_TR, NCR_WID}.

The typical behavior of a RIS is that an impinging signal from a given direction is reflected to an arbitrary, but configurable, outgoing direction; this is commonly termed as an anomalous reflection. In other words, the RIS selects a spatial pair of directions and establishes a strong link between said two directions. This brings us closer to the topic under study in this paper, namely, whether it is true to say that a RIS itself is \emph{spatially selective}? With a spatially selective RIS, we mean one where the \emph{only} reflected signal is the one impinging from the configured input direction. This is, to the best of the authors' knowledge, not the case for any typical RIS technology. Rather, a RIS reflects signals impinging from almost \emph{any} direction (exactly which will be clarified by Lemma 1 and surrounding text), and this fact is of some concern to us as it may cause troublesome interference situations among cells. 

The situation is discussed in \cite[Section 6]{ETSI_GR3}, where the consequences of a RIS reflecting signals from other cells are elaborated upon. For example, significant SINR degradation may occur~\cite[Figure 15]{ETSI_GR3} in some scenarios. Another recent discussion is~\cite{bjornson}, where it is observed that RISs reflecting signals impinging from non-configured directions may cause considerable pilot contamination problems. A solution offered in~\cite{bjornson} is to facilitate coordination among cells during pilot transmissions. However, this solution does not scale well as the number of cells increases.

Altogether, RIS is a spatially non-selective technology, and it is an open question how it can be made spatially selective. In this paper, we shall demonstrate that by re-radiating signals from antenna elements other than the ones they arrive at, a RIS naturally becomes spatially selective. We note that a similar proposal was made in~\cite{clerckx}. However, the focus in~\cite{clerckx} is on maximizing the channel gain, and spatial selectivity was never studied. Our paper is organized as follows: 
Section~\ref{rismodel} lays down a system model for a standard RIS, which Section~\ref{risnonsel} shows spatially non-selective. Section~\ref{RISsel}, our main contribution, proposes a RIS architecture that is spatially selective, and proves this fact analytically. Section~\ref{variants} discusses some variants and properties of the proposed architecture, and Section~\ref{numerical} provides numerical results. Section~\ref{summ} summarizes the paper.

\section{RIS system model} \label{rismodel}
We study a RIS with $M^2$ elements arranged according to a square grid with an inter-element spacing of~$\lambda/2$. Each element phase rotates its impinging signal, and then re-radiates the phase-rotated signal. In this paper, we are exclusively concerned with far-field signals, and we shall use directional cosines, $\mathbf{k}=(k_\mathrm{x},k_\mathrm{y})$ to parameterize spatial directions as seen from the RIS. For a directional cosine to be valid, it must hold that $k_\mathrm{x}^2+k_\mathrm{y}^2\leq 1$, a.k.a. the visible region. In this paper, all directions are defined as outgoing from the RIS. With that, if a signal impinges the RIS from a direction $\k$, then the received signal power in a direction $\kt$, where path-losses, RIS insertion losses, effective aperture losses, etc., have been ignored, reads
\begin{equation} \label{eq1}
A(\k,\kt;\mathbf{C})=|\mathbf{s}(\kt)\circ\mathbf{C}\circ \mathbf{s}(\k)|^2,
\end{equation}
where "$\circ$" denotes the Hadamard product, $\mathbf{C}$ is an $M\times M$ matrix whose $(m,n)$th element represents the phase shift at RIS element $(m,n)$, and $\mathbf{s}(\cdot)$ is an $M\times M$ matrix representing a steering vector associated with the direction specified by its argument. More specifically, element $(m,n)$ of $\mathbf{s}(\k)$ reads
\begin{equation} \label{eq2}
s_{m,n}(\k)=\mathrm{exp}\left(\jmath\phi\right) \mathrm{exp}\left(\jmath \pi (m\kx+n\ky)\right),
\end{equation}
where $\phi$ is a reference phase irrelevant for our purposes and will be left out throughout the paper.

It is much more common in contemporary literature to work with a vectorized version of~\eqref{eq1}. The reader may choose his/her favorite order of vectorization, but we shall express~\eqref{eq1} throughout the paper as
\begin{equation} \label{eq3}
A(\k,\kt;\mathbf{C})=|\mathbf{s}\T(\kt)\mathbf{C} \mathbf{s}(\k)|^2,
\end{equation}
where $\mathbf{C}$ is an $M^2 \!\times \!M^2$ diagonal matrix, and $\mathbf{s}(\cdot)$ are $M^2\!\times\! 1$ vectors.
If we let~$\mathbf{c}$ denote the diagonal elements of $\mathbf{C}$, then we may rewrite~\eqref{eq3} as
\begin{equation} \label{eq4}
A(\k,\kt;\mathbf{c})=|\mathbf{c}\T (\mathbf{s}(\k)\circ\mathbf{s}(\kt))|^2.
\end{equation}
Maximizing $A(\k,\kt;\mathbf{c})$ over $\mathbf{c}$ is trivial and the solution reads
\begin{equation}
   \label{optsol}
    \mathbf{c}^\mathrm{opt}(\k,\kt) = \mathbf{s}(-(\k+\kt)).
\end{equation}
This yields
$$A^\mathrm{opt}(\k,\kt)=A(\k,\kt;\mathbf{c}^\mathrm{opt}(\k,\kt))=M^4.$$
We remark that although $A^\mathrm{opt}(\k,\kt)$ is independent of its arguments, we keep these as they will be useful in later discussions. 
But we shall, for the sake of convenience, drop the superscript "$\mathrm{opt}$" from $\mathbf{c}^\mathrm{opt}(\k,\kt))$ and $A^\mathrm{opt}(\k,\kt)$, and use the general notation $\mathbf{c}(\mathbf{r}) = \mathbf{s}(-\mathbf{r})$: This is the optimal RIS configuration for directions $\k,\kt$ fulfilling  $\k+\kt+\mathbf{r}=\mathbf{0}$.

As a final remark, which will be important later, we note that a RIS satisfies
\begin{equation} \label{recip}
    A(\k,\kt;\cdot)=A(\kt,\k;\cdot)
\end{equation}
no matter what the RIS configuration is. We refer to this as \textit{a RIS is reciprocal}. This is of strong operational significance, as it implies that the same configuration can be used for both uplink and downlink.

\section{A RIS is not spatially selective} \label{risnonsel}
If the RIS is configured with phase configuration $\mathbf{c}(-\k-\kt)$, the impinging signal arrives from direction $\k$, but we measure the received power in a direction $\kt+\delt$, then it can be verifed that
\begin{equation} \label{loss1}
    A(\k,\kt+\delt;\mathbf{c}(-\k-\kt))\ll  A(\k,\kt).
\end{equation}
Likewise, if the measurement is maintained in direction $\kt$, but the signal impinges on the RIS from a direction $\k+\de$, we have
\begin{equation} \label{loss2}
    A(\k+\de,\kt;\mathbf{c}(-\k-\kt))\ll  A(\k,\kt).
\end{equation}
The left-hand-sides of \eqref{loss1}--\eqref{loss2} behave, for random realizations of $\de,\delt$ and as $M\to \infty$, as the powers of the sum of $M^2$ random and independent phasors. Wherefore, the expected values of the left-hand-sides are $M^2$.

Now, the above considerations are perhaps not very exciting, but a more interesting question is the following: For a given RIS configuration $\cb(-\k-\kt)$ and impinging direction $\k+\de$, can we find a measurement direction $\kt+\delt$ such that
\begin{equation} \label{paras}
    A(\k+\de,\kt+\delt;\cb(-\k-\kt))=A(\k,\kt) ?
\end{equation}
To answer this, we make use of Lemma 1.
\begin{lemma}
For any impinging direction $\mathbf{t} = (t_\mathrm{x}, t_\mathrm{y})$, $\|\mathbf{t}\|^2 \leq 1$, and RIS configuration setting $\cb(\mathbf{r})$ there exists a measurement direction $\mathbf{z} = (z_\mathrm{x}, z_\mathrm{y})$ such that $A(\mathbf{t},\mathbf{z};\cb(\mathbf{r}))=M^4$ if, and only if, the following system of equations are solvable
\begin{align*}
    \mathbf{t} + \mathbf{z} + \mathbf{r} = \mathbf{0} \mod{2}\quad\mathrm{s.t.}\ \|\mathbf{z}\|^2\leq 1.
\end{align*} 
\end{lemma}
\begin{proof}
    We have
    \begin{eqnarray} \label{lemeq1}
        A(\mathbf{t},\mathbf{z};\cb(\mathbf{r}))&=& \left|\cb\T(\mathbf{r})\left(\mathbf{s}(\mathbf{t})\circ \mathbf{s}(\mathbf{z}) \right)\right|^2 \nonumber \\
        &=& \left|\cb\T(\mathbf{r})\mathbf{s}(\mathbf{t}+\mathbf{z})\right|^2. 
    \end{eqnarray}
    From (\ref{eq2}) we can express (\ref{lemeq1}) as
\begin{align*}    A(\mathbf{t},\mathbf{z};\cb(\mathbf{r}))&=\Bigl|\sum_{m=1}^M \mathrm{exp}\left(\jmath \pi m(z_\mathrm{x}+t_\mathrm{x}+r_\mathrm{x})\right)\Bigr. \\
& \quad \times \sum_{n=1}^M \mathrm{exp}\left(\jmath \pi n(z_\mathrm{y}+t_\mathrm{y}+r_\mathrm{y})\right)\Bigr|^2.
\end{align*}
To reach $A(\mathbf{t},\mathbf{z};\cb(\mathbf{r}))=M^4$ we must have that each of the constituent sums add up to $M$. But this can only happen if $z_\mathrm{x}+t_\mathrm{x}+r_\mathrm{x} \;\mathrm{mod}\; 2=0$ and $z_\mathrm{y}+t_\mathrm{y}+r_\mathrm{y} \;\mathrm{mod}\; 2=0$.
However, if no such solution exists such that $z_\mathrm{x}^2+z_\mathrm{y}^2\leq 1$ exists, then the direction $\mathbf{z}$ does not belong to the visible region.
\end{proof}

To answer our above-posed question, we may simply identify $\mathbf{r}=-\k-\kt$, $\mathbf{t}=\k+\de$, and $\mathbf{z}=\kt+\delt$. Thus, we may find $\delt$ through 
$$-\k-\kt+\k+\de+\kt+\delt=\de+\delt=\mathbf{0},$$
which gives $\delt=-\de$. If this $\delt$ also satisfies $\|\kt+\delt\|^2\leq 1$, then the direction $\kt+\delt$ belongs to the visible region, and we have found a measurement direction in which the full beamforming gain $M^4$ is obtained. If, on the other hand,  $\|\kt+\delt\|^2> 1$, we may still find a solution by considering equations of the form
$\de+\delt=(p_1,p_2)$ with $p_k\in\{0,\pm 2\}$. If any of these equations produces a $\kt+\delt$ in the visible region, we have found a measurement direction with full beamforming gain. 
It is not hard to find a case where no such measurement direction exists. For example, selecting $\k+\de+\mathbf{r}=(0.8,0.8)$ where $\cb(\mathbf{r})$ is the RIS configuration, has no such  direction. 

However, for a given RIS configuration $\cb(\mathbf{r})$ and a randomly selected impinging direction, the probability that a measurement direction for which the full gain $M^4$ is achieved exists is high. This may, potentially, lead to cumbersome interference and RIS-coexistence situations. We give an illuminating example, based on \cite[Section 6]{ETSI_GR3}, next.
\begin{example}
Assume two neighbor systems that share frequency bands. System A is RIS-empowered, and its RIS has been configured. 
However, in the downlink of system B, the base station's (BS) signal reaches the RIS of system A, and the particular RIS configuration is such that system B's user is receiving the RIS-reflected signal. The user in system B has estimated its downlink channel based on pilot signals, and the channel estimate therefore includes the RIS-induced link. Assume now that the RIS in system A is reconfigured. This is an abrupt event, and the channel to the user of system B is instantaneously impacted. Such channel change is far too rapid to be tracked by pilot signals, and this leads to severely degraded performance.
\end{example}

We refer to the situation discussed in this section as \textit{a RIS is not spatially selective}, as a RIS may potentially reflect any incoming signal, with full beamforming gain, somewhere in space. This is not especially surprising, as a good analogy for a RIS is a standard mirror. A RIS reconfiguration may be considered as a tilt of the mirror. For a mirror, it is absolutely clear that light from any direction will be reflected---and this carries over to RISs, save for the fact that the equations in Lemma 1 must be solvable.

\section{A spatially selective RIS implementation} \label{RISsel}
In this section, we put forth a RIS implementation that only reflects signals from one given direction. By contemplating upon how this may possibly be done, one quickly reaches the conclusion that it cannot be done through means of ordinary signal processing, i.e., by adapting the phase configuration $\cb$. Rather, the RIS hardware must be re-engineered in such a way that the RIS is spatially selective. 

In standard RIS technology, the signal received at a certain element is first phase rotated and then re-radiated from the very same element. If this is altered, i.e., the phase-rotated signal leaves the RIS from an element other than the one receiving it, the RIS naturally becomes spatially selective. 

Let~$\sigma$ denote a permutation of the integers $1,2,\ldots,M^2$, and let $\mathbf{P}$ denote a permutation matrix constructed from~$\sigma$. That is, the $(i, j)$ entry of $\mathbf{P}$ is 1 if $\sigma(j) = i$, and 0 otherwise. If we assume that the phase rotation takes place at the receiving antenna, then \eqref{eq3} is replaced with 
\begin{equation} \label{rew_eq3}
A(\k,\kt;\mathbf{C})=|\mathbf{s}\T(\kt)\mathbf{P}\mathbf{C} \mathbf{s}(\k)|^2.
\end{equation}
Despite the permutation, it is straightforward to solve~\eqref{rew_eq3} for the optimal phase configuration $\mathbf{C}$. If its main diagonal, $\cb$, is configured to $\cb=\cb(\sigma,\kt,\k)$ where
$$\cb(\sigma,\kt,\k)=(\sigma^{-1}(\mathbf{s}(\kt))\circ \mathbf{s}(\k))^\ast$$
it is easily verified that 
$$A(\k,\kt;\cb(\sigma,\kt,\k))=M^4,$$
that is, the full beamforming gain is reached.

With that, we next turn our attention to impinging signals from other directions than $\k$ and consider whether we may find a direction $\kt+\delt$ for which 
$A(\k+\de,\kt+\delt;\cb(\sigma,\kt,\k))=M^4$. This is answered in Theorem 1.
\begin{theorem} \label{thm1}
    For separable permutations, i.e., element $(m,n)$ is permuted to element $(\sigma(m),\sigma(n))$, there exists $\sigma$ such that the only solution to $A(\k+\de,\kt+\delt;\cb(\sigma,\kt,\k))=M^4$ is $\de=\delt=\mathbf{0}$ if, and only if, $M\geq 4$.
\end{theorem}
The proof of Theorem 1 is provided in Appendix A. The theorem merely states that a permutation exists so that the RIS is spatially selective. However, practically finding a permutation fulfilling the condition in the theorem is trivial---any randomly chosen $\sigma$ will essentially do - separable or not. Permutations with strong structure, such as if element $(m,n)$ is mapped to $(M+1-m,M+1-n)$, are exceptions. A more interesting question is to what extent the RIS can be made spatially selective through permutations, i.e., what is the maximum of $A(\k+\de,\kt+\delt;\cb(\sigma,\kt,\k))$
for $\de,\delt$ "not close to zero". We remark that for non-separable permutations, the condition $M\geq 4$ needs not to hold, and, moreover, separable permutations of the form $(m,n)\to (\sigma_1(m),\sigma_2(n))$ still requires $M\geq 4$ (which an astute reader can observe from the proof, but not pointed out).
This discussion is deferred to Section \ref{numerical}. But before we tend to that, we  investigate a number of other practical aspects.

\section{Aspects of spatially selective RISs} \label{variants}

\subsection{Reciprocity}
Theorem~\ref{thm1} states that, for permutations of practical interest, the only solution to $A(\k+\de,\kt+\delt;\cb(\sigma,\kt,\k))=M^4$ is $\de=\delt=\mathbf{0}$; this was also pointed out in~\cite{clerckx}. Specifically, this implies that 
$A(\kt,\k;\cb(\sigma,\kt,\k))<M^4,$
that is, the RIS is not reciprocal (unless, of course, $\k=\kt$, i.e., retro-reflections). This implies that the RIS needs to be re-configured for uplink transmissions in case the configuration $\cb(\sigma,\kt,\k)$ is made for the downlink (and vice versa). Yet another way to see this is to outright observe that, in general,
$\mathbf{x}\T\mathbf{P\mathbf{y}}\neq\mathbf{y}\T\mathbf{P\mathbf{x}}$, $\mathbf{P}$ being a permutation matrix. Taking $\mathbf{y}=\mathbf{C}\mathbf{s}(\k)$ and $\mathbf{x}=\mathbf{s}(\kt)$ shows that the RIS is no longer reciprocal. 

To solve this issue, we may rely on beam splitting. Specifically, we configure the RIS with the phase setting
\begin{equation} \label{beamsplit}
\cb(\sigma,\k\leftrightarrow\kt)=\frac{\cb(\sigma,\kt,\k)+\cb(\sigma,\k,\kt)}{|\cb(\sigma,\kt,\k)+\cb(\sigma,\k,\kt)|},
\end{equation}
where both the division and magnitude operator should be understood elementwise.
The performance of such a setting is clarified in
\begin{lemma}
    As $M\to\infty$, and except for a set of directions $\{\k,\kt\}$ with measure 0, 
    \begin{align*}
    \lim_{M\to\infty}\!\!\!\frac{A(\k,\kt;\cb(\sigma,\k \!\leftrightarrow\!\kt))}{M^4}=\!\!\!\lim_{M\to\infty}\!\!\!\frac{A(\kt,\k;\cb(\sigma,\k \!\leftrightarrow\!\kt))}{M^4}\!=\!\frac{4}{\pi^2}. \nonumber 
    \end{align*}    
\end{lemma}
\begin{proof}
It is easily seen that the phase of an individual element of $\cb(\sigma,\k \leftrightarrow\kt))$ is exactly the phase in between that of its constituent parts $\cb(\sigma,\kt,\k)$ and $\cb(\sigma,\k,\kt)$. The only ambiguity occurs if the said two constituent phases differ by precisely $\pi$; but this is an event with measure 0 and, thus, disregarded.

We may express, e.g., $A(\k,\kt;\cb(\sigma,\k \leftrightarrow\kt))$ as
\begin{equation} \label{lemma2:eq1}
A(\k,\kt;\cb(\sigma,\k \leftrightarrow\kt))=\Bigl |\sum_{m=1}^{M^2}c_mv_m(\sigma,\k,\kt)\Bigr|^2,
\end{equation}
where $c_m$ are the diagonal elements of $\mathbf{C}$ (fully specified by (\ref{beamsplit})), and $\mathbf{v}(\sigma,\k,\kt)=[v_1(\sigma,\k,\kt),\ldots,v_{M^2}(\sigma,\k,\kt)]$ may be explicitly found from $\sigma$, $\k$, and $\kt$.

But for arbitrarily selected directions, $\k$ and $\kt$ are irrational with probability one, and as $M\to\infty$, the phases $v_m(\sigma,\kt,\k)$ and $v_m(\sigma,\kt,\k)$ are statistically independent and uniformly distributed over $[0,2\pi)$. With that, the $c_m$ may be regarded as uniformly distributed over $[\angle v_m^\ast(\sigma,\k,\kt)-\pi/2,\angle v_m^\ast(\sigma,\k,\kt)+\pi/2)$. Consequently, we can express (\ref{lemma2:eq1}), in the limit as $M\to \infty$, as
\begin{align*}
\lim_{M\to\infty}\frac{1}{M^4}\sum_{m=1}^{M^2}\bigl|c_mv_m(\sigma,\k,\kt)\bigr|^2 = \Bigl|\frac{1}{\pi}\int_{-\pi/2}^{\pi/2}e^{\jmath x}\mathrm{d}x\Bigr|^2 =\frac{4}{\pi^2}.
\end{align*}
\begin{figure*}[t]
    \centering
    \includegraphics[width=.82\linewidth]{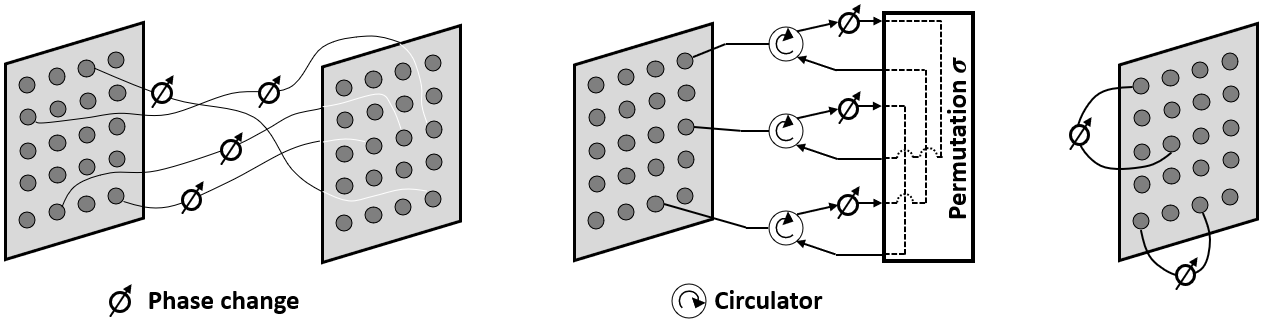}
    \caption{Left: A two-panel permuted RIS. Middle: A one-panel permuted RIS with circulators. Right: A one-panel permuted RIS with a symmetric permutation. Not all connections are shown for improved visibility.}
    \label{fig:ris_hws}
\end{figure*}
By symmetry, it follows that, in the limit as $M\to \infty$, also $A(\kt,\k;\cb(\sigma,\k \leftrightarrow\kt))/M^4\to 4/\pi^2$.
\end{proof}
The implication of Lemma 2 is that requiring reciprocity exacts a $\approx\! 4\,\mathrm{dB}$ power loss. Finally,   the power losses in the uplink and downlink are subject to a tradeoff by replacing~\eqref{beamsplit} with the weighted version 
\begin{equation} \label{beamsplit2}
\cb(\sigma,\kt\leftrightarrow\k)=\frac{\alpha\cb(\sigma,\kt,\k)+(1-\alpha)\cb(\sigma,\k,\kt)}{|\alpha\cb(\sigma,\kt,\k)+(1-\alpha)\cb(\sigma,\k,\kt)|}.
\end{equation}

\subsection{Hardware implementations}
We next turn to hardware architectures for RISs employing element permutations. We make a distinction between one-panel and two-panel RIS, as elaborated upon next.

\subsubsection{Two-panel RIS implementations}
In some scenarios, a RIS is favorably implemented using two antenna arrays. Examples may include indoor-to-outdoor, coverage extension, etc. In principle, one array is mounted to face the BS, while the second array faces the users. The two panels may be mounted back-to-back, but in indoor-to-outdoor scenarios, they may be separated by, e.g., a concrete wall. For a two-panel RIS, a permutation among elements is particularly easy to implement. 

Consider two $M\times M$ arrays. Without any permutation, a RIS would connect the $m$th element of both panels via a configurable phase shifter; here, the $m$th element is arbitrary as long as the order of elements is indexed in the same way at both panels. However, a spatially selective RIS based on a permutation connects element $m$ of, say, the first panel with element $\sigma(m)$ at the second panel, again via a phase shifter. This is illustrated in Figure \ref{fig:ris_hws} (left).

\subsubsection{One-panel RIS implementations} \label{sec:one_panel}
For the typical one-panel RIS, the permutation is more cumbersome to implement. The most straightforward way appears to be through circulators, as shown in Figure \ref{fig:ris_hws} (middle). The box labeled "permutation~$\sigma$" implements the permutation, and this may be fixed at design, or reconfigurable. Reconfigurable permutations are certainly of interest and may bring improved spatial selectivity, but are not studied any further in this paper.

\subsection{Symmetric permutations}
The use of circulators in Section~\ref{sec:one_panel} unarguably complicates the hardware, and it would be beneficial if circulators could be dispensed with. This is indeed possible if the permutation is symmetric, i.e., $\mathbf{P}=\mathbf{P}\T$. For symmetric permutations, an implementation is shown in Figure~\ref{fig:ris_hws} (right). 
As is apparent from the figure, only a single phase shifter is now present for each pair of elements. This will lead to a loss in beamforming gain, but avoids the need of any circulators. We quantify this loss in Theorem \ref{thm2} (whose proof appears in Appendix B).

\begin{theorem} \label{thm2}
For a random but symmetric permutation $\sigma$, and with randomly selected directions $\k, \kt$, except for those in a set of measure 0, we have, as $M\to \infty$, that 
$$\frac{A(\k,\kt;\sigma)}{M^4}\to \frac{4}{\pi^2}.$$
Moreover, a symmetric permutation ensures a reciprocal RIS.
\end{theorem}

It is interesting to observe that the result in Theorem \ref{thm2} coincides with that of beam splitting. Yet, the symmetric property avoids any use of circulators which is a major simplification. Finally, we reiterate that the loss of the symmetric permutation is due to the symmetric property, but rather a consequence of the absence of circulators. If circulators are used, then a symmetric permutation suffers no loss in beamforming gain as the number of phase shifters equals the number of elements. However, this comes at the expense of reciprocity; to avoid the loss in beamforming gain, different paths must have different phase shifter settings, which implies that reciprocity is lost.


\section{Numerical results} \label{numerical}
Theorem \ref{thm1} merely states that a permutation exists such that the full beamforming gain is not achieved in any other pair of directions than the configured pair. This is, however, a grossly restrained statement, and in practice any randomly chosen permutation will suffice. More interesting questions are, e.g., (i) To what extent is the RIS spatially selective? i.e., what is the reduction in beamforming gain for non-configured directions?, (ii) How difficult is it to find a good permutation?, (iii) Does the permutation significantly alter the shape of the main lobe? (iv) What is the gain of having a reconfigurable permutation, i.e., one that is not fixed at fabrication?

A full investigation of these, and similar, questions is not possible within the scope of this paper; wherefore we aim at providing the reader with an initial understanding of how permutations impact RIS reflections. To determine a small set of permutations that can be considered "optimal/near-optimal" with respect to a set of metrics is of huge practical importance, but left for future research. Furthermore, (iv) is left for future research, as well as a structured methodology for (ii).

\subsection{Volume of main lobe}
Using~\eqref{interm} of Appendix A, the beamforming gain equals
$$A(\k+\de,\kt+\delt;\cb(\sigma,\kt,\k))=A(\de,\delt;\cb(\sigma,\mathbf{0},\mathbf{0})),$$
i.e., it is independent of $\k,\kt$. The volume of main lobe of the RIS reflection pattern may be loosely defined as the volume of the set ${\de,\delt}$ for which $A(\de,\delt;\cb(\sigma,\mathbf{0},\mathbf{0}))\approx M^4$. The shape of said set may be complicated and therefore cumbersome to treat analytically, and, further, as there are four variables involved, not easy to illustrate graphically. As a remedy, we make the constraint $(\|\de\|^2+\|\delt\|^2)^{1/2}\leq \delta$, and define
$$\beta=\min_{(\|\de\|^2+\|\delt\|^2)^{1\!/2} \,\leq\, \delta}\frac{1}{M^4}A(\de,\delt;\cb(\sigma,\mathbf{0},\mathbf{0})).$$
This corresponds to the smallest value of the beamforming gain for directions $\de,\delt$ in a 4-ball of radius $\delta$. In Figure \ref{fig:mainlobe} we plot $\beta$ against $\delta$ for 100 randomly selected permutations; we consider three different values of $M$, namely, $5,10$ and $20$. The blue and the red curves highlight the top and bottom curves, respectively. The green curves are the ensuing results for the identity permutation, i.e., a standard RIS. We can draw the following conclusions:
\begin{itemize}
    \item A system designer has freedom in the design of the main lobe: By carefully selecting the permutation, one may choose a permutation with a "widest" (top, blue curve) and "narrowest" (bottom, red curve) main lobe.
    \item As $M$ grows, there is less and less difference among permutations, and it appears that any randomly selected permutation suffices. This is most natural as average behaviors across the RIS are dominant. 
    \item In general, the shape of the main lobe is not drastically changed compared with that of the identity permutation, but it appears to widen slightly.
\end{itemize}
\begin{figure}
    \centering
    \hspace*{-4mm}\includegraphics[width=\linewidth]{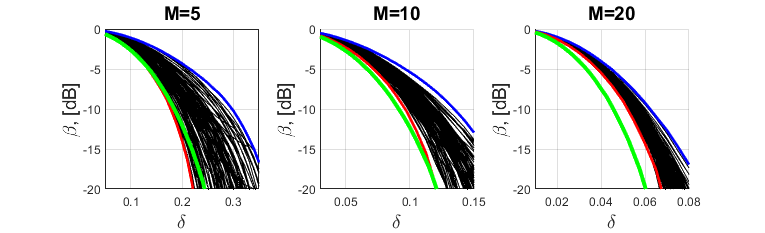}
    \caption{Behavior of the main lobe for 100 randomly selected permutations (black curves). The green curves show the results of the identity permutation. The blue and red curves show the top and bottom of the black curves.}
    \label{fig:mainlobe}
\end{figure}
The optimization involved in finding $\beta$ is not trivial, and we have used Matlab's built-in solvers initiated with $10^3$ random starting positions. 

\subsection{Degree of spatial selectivity}
The whole point of the permutation is to ensure that $A(\de,\delt;\cb(\sigma,\mathbf{0},\mathbf{0}))\ll M^4$ unless $\de,\delt$ are close to zero. So far, we have not quantified to what extent this is true, and we turn to this, most important, question next. In principle, we are interested in $\max_{\de,\delt} A(\de,\delt;\cb(\sigma,\mathbf{0},\mathbf{0}))$, and would like to find $\sigma$ such this is as small as possible. However, for $\de,\delt$ close to zero, this quantity will naturally be large, as we are within the main lobe. Wherefore, we constrain $\de,\delt$ to lie outside a 4-ball of radius $\delta$ centered at the origin.  Suitable values of $\delta$ can be understood from Figure \ref{fig:mainlobe}, and we here choose $\delta\in\{ 0.3,0.15,0.08\}$ for $M\in\{5,10,20\}.$ With that we find (through numerical optimization initiated with $10^3$ starting positions)
$$\tau=\max_{(\|\de\|^2+\|\delt\|^2)^{1\!/2} \,>\, \delta}\frac{1}{M^4}A(\de,\delt;\cb(\sigma,\mathbf{0},\mathbf{0})).$$

The results are shown as empirical CDFs in Figure \ref{fig:sidelobe} generated from the same set of 100 randomly generated permutations as in Figure \ref{fig:mainlobe}. We remark that the identity permutation, corresponding to the green curves in Figure \ref{fig:sidelobe}, leads to $\tau=1$ as such RIS is not spatially selective. 
\begin{figure}
    \centering
    \includegraphics[width=\linewidth]{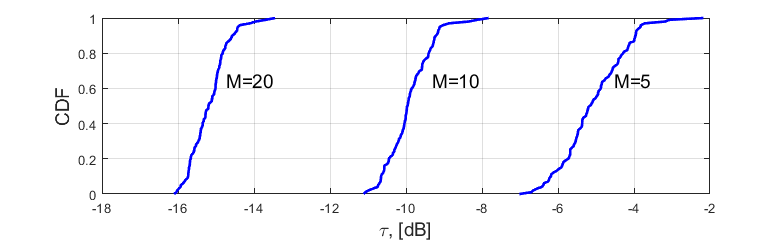}
    \caption{Empirical CDFs of $\tau$ for 100 randomly selected permutations.}
    \label{fig:sidelobe}
\end{figure}


\subsection{Some example RIS reflection patterns}
To graphically visualize $A(\de,\delt;\cb(\sigma,\mathbf{0},\mathbf{0}))$ is difficult as there is a total of four variables in $\k,\kt$. Therefore, we assume $\de=\rho_\mathrm{x}[1\;1]$ and $\delt=\rho_\mathrm{x}[1\;1]$ and illustrate $A(\de,\delt;\cb(\sigma,\mathbf{0},\mathbf{0}))/M^4$ as function of $\rho_\mathrm{x}$ and $\rho_\mathrm{y}$. We set $M=10$ and show the results in Figure \ref{numex} for two randomly selected permutations, as well as for the identity permutation.
\begin{figure}
    \centering
    \includegraphics[width=\linewidth]{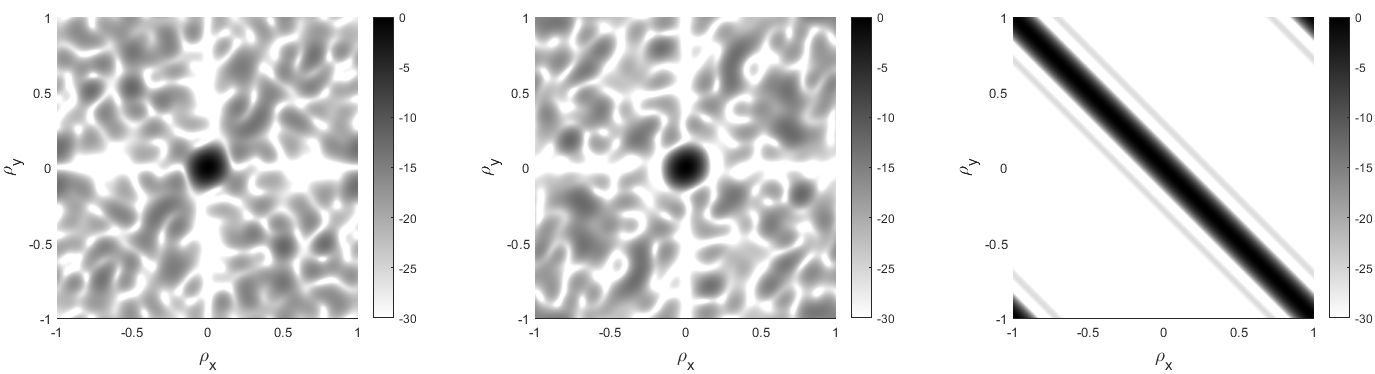}
    \caption{Example reflection patterns $A(\de,\delt;\cb(\sigma,\mathbf{0},\mathbf{0}))/M^4$ with $\de=\rho_\mathrm{x}[1\;1]$ and $\delt=\rho_\mathrm{x}[1\;1]$. Leftmost and center: Random permutations. Rightmost: Identity permutation.}
    \label{numex}
\end{figure}
That a standard RIS is spatially non-selective is seen in the rightmost plot, as no matter the impinging direction, there is an outgoing direction with full beamforming gain.

%

\section{Summary} \label{summ}
In this paper, we first observed that a standard RIS is not spatially selective, i.e., it reflects signals from almost any impinging direction. We provided an example of why this may be of concern and then put forth a RIS architecture that is spatially selective. Said architecture is based on a RIS in which the signal received at a certain antenna is re-radiated from a different antenna; this was proven to enforce spatial selectivity. We then proceeded by discussing hardware implementations and RIS reciprocity, i.e., the RIS should be capable of reflecting in both uplink and downlink without any reconfiguration. Reciprocity can be furnished for, but comes at the price of a $\approx 4\,\mathrm{dB}$ beamforming gain penalty. We also provided a numerical study illustrating the degree of spatial selectivity offered at various RIS dimensions.

\section*{Appendix A: Proof of Theorem 1}
The proof is constructive; it provides $\mathbf{\sigma} = (\sigma(m),\sigma(n))$ fulfilling Theorem 1. 
Let $\mu=A(\k+\de,\kt+\delt;\cb(\sigma,\kt,\k))$, $\mathbf{r} = (m,n)$, and let $\langle \cdot \; , \cdot \rangle$ denote the scalar product. Then, 
\begin{align} 
 \mu &= \Bigl|\sum_{m,n}\exp{\!\bigl(\jmath\pi (\langle\mathbf{r},\k \!+\! \de\rangle + \langle\mathbf{\sigma},\kt \!+\! \delt\rangle - \langle\mathbf{r},\k\rangle - \langle\mathbf{\sigma},\kt\rangle \bigr)} \Bigr|^2\nonumber \\ 
&= \Bigl |\sum_{m,n}\exp{\!\bigl(\jmath\pi (\langle\mathbf{r},\de \rangle + \langle\mathbf{\sigma}, \delt\rangle \bigr)} \Bigr |^2\nonumber \\ 
&=A(\de,\delt;\cb(\sigma,\mathbf{0},\mathbf{0})) \label{interm} \\ 
&=\mu_{\mathrm{x}}\mu_{\mathrm{x}}, \label{fin:eq0} 
\end{align}
where $\mu_\mathrm{v}=\Bigl|\sum_{m=1}^M\exp{\bigl(\jmath\pi (m\Delta_\mathrm{v}\!+\!\sigma(m)\tilde{\Delta}_\mathrm{v}\!)\bigr)}\Bigr|^2,$ $\mathrm{v}\in\{\mathrm{x},\mathrm{y} \}$.
We remark that (\ref{interm}) holds for general non-separable permutations, whereas (\ref{fin:eq0}) holds only for separable ones.

To obtain $\mu=M^4$, both $\mu_\mathrm{x}$ and $\mu_\mathrm{y}$ must equal $M^2$. We may wlog focus on $\mu_\mathrm{x}$, say, and consider sufficient conditions for  $\mu_\mathrm{x}=M^2$. Define $d_m=\sigma(m)-m$. We  have
\begin{align} \label{newproof:eq1}
\mu_\mathrm{x}&=\Bigl | \sum_{m=1}^M\exp{\bigl(\jmath\pi m(\Delta_\mathrm{x}+\tilde{\Delta}_\mathrm{x})\bigr)}\exp{\bigl(\jmath\pi d_m\tilde{\Delta}_\mathrm{x}\bigr)} \Bigr | ^2 \nonumber \\
&= \Bigl |\sum_{m=1}^M\exp{\bigl(\jmath 2\pi m f\bigr)}u_m \Bigr |^2,
\end{align}
where $f\!=\!(\Delta_\mathrm{x}\!+\!\tilde{\Delta}_\mathrm{x})/2$ and $u_m\!=\!\exp{(\jmath\pi d_m\tilde{\Delta}_\mathrm{x})}$. The sum in (\ref{newproof:eq1}) is the Discrete time Fourier transform $U(-f)$ of the sequence $u_m$. Therefore, $\mu_\mathrm{x}=M^2$ requires that $|U(f)|^2$ attain its upper bound $\sum |u_m|^2=M^2$ at some $f$. But this only happens if $u_m\!=\!\exp{(\jmath(\theta+\!m\tilde{f}\!+\!2\pi p_m))}$, with $\theta\in[0,2\pi)$ and $p_m$ an integer. 
Specifically, with integer-valued $q_m$
$$(\angle u_2-\angle u_1)+2\pi q_1 =  (\angle u_3-\angle u_2)+2\pi q_2,$$
which implies, after a straightforward manipulation, that $$\tilde{\Delta}_{\mathrm{x}}=\frac{2(q_2-q_1)}{2d_2-d_1-d_3}.$$
Let $M=3$. Then, $2d_2-d_1-d_3$ must be 0 or $\pm 3$. If $2d_2-d_1-d_3=0$, the permutation must be either the identity or the permutation $(3,2,1)$. For the latter, it can be verified that spatial selectivity does not hold; see Sec.~\ref{risnonsel}. If $2d_2-d_1-d_3=\pm 3$, we may choose  $q_2-q_1=2$ so that $\tilde{\Delta}_{\mathrm{x}}=\pm 2/3$. Setting $f=\mp 2/3$ gives $\Delta_\mathrm{x}=\pm 2$, which is congruent to 0 modulus 2, and we may take $\Delta_\mathrm{x}=0$. Repeating this for other values of $q_2-q_1$, and for $\mu_\mathrm{y}$ shows that there exist $\k+\de\neq\k$ and $\kt+\delt\neq\kt$ within the visible region such that the full beamforming gain is achieved. Thus, no $\sigma$ can yield a spatially selective RIS.

Consider, for~$M\geq 4$, the permutation $(4,3,1,2,5,\ldots,M)$. Note that~$d_1=3$, $d_2=1$, $d_3=-2$, $d_4=-2$, $d_5=\cdots=d_M=0$, and so $2d_2-d_1-d_3=1$. Hence, $\mu_\mathrm{x}=M^2$ implies $\tilde{\Delta}_\mathrm{x}=2(q_2-q_1)$. But this is an even integer, i.e., congruent to 0 modulus 2. Thus, $u_m=1$ is the only possible solution, and it follows that $f$ must be integer-valued. This, in turn, implies that $\Delta_\mathrm{x}$ is congruent to 0 modulus 2. Repeat the argument for $\mu_\mathrm{y}$ to obtain $\tilde{\Delta}_\mathrm{y}=\Delta_\mathrm{y}=0\mod{2}.$ We have constructed a permutation with $M\geq 4$ such that the RIS is always spatially selective. 

\section*{Appendix B: Proof of Theorem \ref{thm2}}
Let elements $(m,n)$ and $(m^\prime,n^\prime)$ be connected. 
The total phase change, save for the contribution of the phase shifter, for the signal received at element $(m,n)$ is
\begin{equation*}
\mathrm{exp}(\jmath \phi_{m,n}\!)=\mathrm{exp}\left(\jmath\pi(mk_\mathrm{x}+nk_\mathrm{y})\right)\mathrm{exp}(\jmath\pi(m^\prime\tilde{k}_\mathrm{x}+n^\prime\tilde{k}_\mathrm{y})),
\end{equation*}
while that of the received signal at element $(m^\prime,n^\prime)$ reads
\begin{align}
\mathrm{exp}(\jmath \phi_{m^\prime,n^\prime})&= \mathrm{exp}(\jmath\pi(m^\prime k_\mathrm{x}\!+\!n^\prime k_\mathrm{y}))\,\mathrm{exp}(\jmath\pi(m\tilde{k}_\mathrm{x}\!+\!n\tilde{k}_\mathrm{y})) \nonumber \\
&=\mathrm{exp}(\jmath\pi \theta(m,n,\k,\kt))\,\mathrm{exp}(\jmath \phi_{m,n}), \nonumber 
\end{align}
where $\theta(m,n,\k,\kt)=e_{m,n}(k_\mathrm{x}-\tilde{k}_\mathrm{x})+d_{m,n}(k_\mathrm{y}-\tilde{k}_\mathrm{y})$, with $m+e_{m,n}=m^\prime$, and $n+d_{m,n}=n^\prime$.
This gives,
\begin{align*}
A(\k,\kt;\sigma)\!=\!\max_{c_{m,n}} \Bigl|\sum_{m,n}\!c_{m,n}e^{\jmath \phi_{m,n}} \bigl(1\!+\!\mathrm{exp}(\jmath\pi \theta(m,n,\k,\kt))\bigr)\Bigr|^2\!.
\end{align*}
Solving this is trivial, and results in
$$A(\k,\kt;\sigma)=\Bigl(\sum_{m,n}\bigl|1+\mathrm{exp}(\jmath\pi \theta(m,n,\k,\kt))\bigr|\Bigr)^2.$$
For randomly selected directions, $\k, \kt$ have irrational components with probability 1. Further, as $M\to \infty$, the equidistribution theorem kicks in. With a slight adaptation, the equidistribution theorem can be, e.g., using Weyl's condition, tweaked to mod 2 arithmetic. This implies that the term $\theta(m,n,\k,\kt)$ is uniformly distributed over the interval $[0,2)$. With that,
\begin{align*} 
\lim_{M\to \infty}\frac{A(\k,\kt;\sigma)}{M^4}&=\Bigl(\int_{0}^2 \!\sqrt{\frac{1}{8}-\frac{\cos(\pi x)}{8}}\mathrm{d}x \Bigr)^2 = \frac{4}{\pi^2}.
\end{align*}

Finally, reciprocity follows by, for instance, noting that $\theta(m,n,\k,\kt)=-\theta(m,n,\kt,\k)$, hence,  the optimal solution for $c_{m,n}$ remains the same if $\k$ and $\kt$ switch roles.

\end{document}